\documentclass[a4paper,UKenglish,cleveref, autoref, thm-restate, nolineno]{socg-lipics-v2021}

\pdfoutput=1 
\hideLIPIcs  

\usepackage{amsmath,amssymb,amsthm}
\usepackage{mathtools}
\allowdisplaybreaks
\usepackage{array}
\usepackage[table]{xcolor}
\usepackage[linesnumbered,ruled,vlined]{algorithm2e}
\usepackage{tikz}
\usepackage{float}
\usepackage[most]{tcolorbox}



\let\epsilon\varepsilon

\newcommand{\bigO}[1]{\mathcal{O}\!\left(#1\right)}
\newcommand{\tO}[1]{\widetilde{\mathcal{O}}\!\left(#1\right)}

\newcommand{\mset}[1]{\{\!\{#1\}\!\}}
\DeclareMathOperator{\opt}{OPT}

\DeclareMathOperator{\SSE}{SSE}

\DeclareMathOperator{\OPT}{OPT}

\SetKwInOut{Input}{Input}
\SetKwInOut{Output}{Output}

\title{Nearly Optimal Bounds for Computing Decision Tree Splits in Data Streams}

\author{Hoang {Ta}}{Hanoi University of Science and Technology, Hanoi, Vietnam}{hoang.taduy@hust.edu.vn}{https://orcid.org/0009-0008-0808-6466}{Corresponding author.}

\author{Hoa T. Vu}{San Diego State University, San Diego, USA}{hvu2@sdsu.edu}{https://orcid.org/0000-0001-8873-0208}{Supported by NSF Grant No. 2342527. Corresponding author.}

\authorrunning{H. Ta and H.\,T. Vu}

\Copyright{Hoang Ta and Hoa T. Vu}



\EventEditors{Philip Bille, Seth Pettie, and Sabine Storandt}
\EventNoEds{3}
\EventLongTitle{34th Annual European Symposium on Algorithms (ESA 2026)}
\EventShortTitle{ESA 2026}
\EventAcronym{ESA}
\EventYear{2026}
\EventDate{August 31--September 4, 2026}
\EventLocation{L'Aquila, Italy}
\EventLogo{}
\SeriesVolume{388}
\ArticleNo{40}

\begin{document}

\maketitle

\begin{abstract}
We establish nearly optimal upper and lower bounds for approximating decision tree splits in data streams. For regression with labels in the range $\{0,1,\ldots,M\}$, we give a one-pass algorithm using $\tO{M^2/\epsilon}$ space \footnote{$\tO{\cdot}$ hides polylogarithmic factors.} that outputs a split within additive $\epsilon$ error of the optimal split, improving upon the two-pass algorithm of Pham et al.\ (ISIT~2025). Furthermore, we provide a matching one-pass lower bound showing that $\Omega(M^2/\epsilon)$ space is indeed necessary.

For classification, we also obtain a one-pass algorithm using $\tO{1/\epsilon}$ space for approximating the optimal Gini split, improving upon the previous $\tO{1/\epsilon^2}$-space algorithm. We complement these results with matching space lower bounds: $\Omega(1/\epsilon)$ for Gini impurity and $\Omega(1/\epsilon)$ for misclassification (which matches the upper bound obtained by sampling).

Our algorithms exploit the Lipschitz property of the loss functions and use reservoir sampling along with Count--Min sketches with range queries. Our lower bounds follow from careful reductions from the \textsc{Index} problem.

\end{abstract}

\section{Introduction}
\label{sec:introduction}

Decision trees are ubiquitous in machine learning. They are used as base learners for powerful ensemble methods such as random forests~\cite{Ho95}, gradient boosting~\cite{MasonBBF99,Friedman2002}, AdaBoost~\cite{FreundS97}, XGBoost~\cite{ChenG16}, LightGBM~\cite{KeMFWCMYL17}, and CatBoost~\cite{ProkhorenkovaGV18}. These ensemble models achieve state-of-the-art performance on many tasks, particularly on tabular data~\cite{GrinsztajnOV22,Shwartz-ZivA22}. At the algorithmic level, a basic primitive in these methods is \emph{split selection}: given a feature and a loss criterion, choose the threshold that yields the best partition of the data~\cite{BreimanFOS84,Quinlan93}.

In large-scale settings, the training data may arrive continuously as a stream and cannot be stored in memory in its entirety. This makes it natural to ask whether one can approximate the best split using only one or a few passes and sublinear space. Beyond its practical relevance, this question is also of independent theoretical interest since split selection is the core operation repeated throughout tree learning and other data analysis tools.

In this paper, we study this problem in the insertion-only streaming model and establish nearly optimal bounds for several standard split objectives including mean squared error for regression, misclassification rate, and Gini impurity for classification.

\subparagraph{Related Work.} Decision trees in data streams have been extensively studied. Domingos and Hulten~\cite{DomingosH00} introduced VFDT, the seminal decision tree learning algorithm for classification for an infinite i.i.d.\ stream; regression is not addressed in their work. A large subsequent literature has studied adaptive random forests~\cite{GomesBRBEPHA17}, concept drift~\cite{WangFYH03, CanoK22}, and other extensions of tree learning in data streams~\cite{hulten2001mining, jin2003efficient, bifet2009new, rutkowski2014cart, manapragada2018extremely}. Others have provided open-source implementations~\cite{BifetZFHZQHP17, Montiel2021} of decision tree learning from data streams, notably the Python packages \textsc{scikit-multiflow} and \textsc{River}.

Recently, Tatti~\cite{Tatti25} studied fast algorithms for finding decision tree splits in sparse data streams, Silva et al.~\cite{SilvaVG25} considered split computation in federated streaming settings, and Assis et al.~\cite{AssisBE25} investigated split selection from the perspective of structural robustness in evolving streams. 

In contrast to the work of Domingos and Hulten~\cite{DomingosH00}, which focused only on classification, we do not assume that the elements in the stream are independently and identically distributed. Another challenge we must address is that, unlike misclassification, regression and Gini-type objectives depend on aggregate statistics over ranges of feature values, and these quantities must be estimated accurately in one pass using sublinear space. Our approach combines two ingredients: reservoir sampling, which yields a small candidate set containing a near-optimal split, and coupled dyadic Count--Min sketches, which maintain the range statistics needed to evaluate the loss. For regression, this gives a truly one-pass improvement over the previous two-pass result of \cite{PhamTV25}; for classification with Gini impurity, a simple label re-encoding reduces the problem to the regression setting.

Since exact decision tree learning is NP-hard \cite{HyafilR76}, greedy heuristics are commonly employed. Top-down frameworks like CART~\cite{BreimanFOS84} and ID3~\cite{Quinlan93} recursively optimize feature thresholds to minimize splitting criteria (e.g., Gini impurity and MSE), which serve as the core computational primitives of tree construction.

\subparagraph{Problem formulation.}
We model the feature domain as $\{1,2,\ldots,N\}$. For real-valued features, this reflects discretization into bins~\cite{BreimanFOS84,ChengFIQ88,Quinlan93,HongLK06,DoughertyKS95,FayyadI92}.

The following formulation of the optimal split problem follows the standard approach of \cite{BreimanFOS84} and \cite{Quinlan93}. For a more recent overview, we refer to Chapter 9 of \cite{HastieTF09}. 

If $k$ is a positive integer, we use $[k]$ to denote the set $\{1, 2, \ldots, k\}$. We denote the indicator variable for the event $Z$ by $\mathbf{1}[Z]$. 

\subparagraph{Regression.} For regression, we are given a stream of pairs
\[
(x_1,y_1),(x_2,y_2),\dots,(x_m,y_m),
\qquad
x_i \in [N],\ y_i \in \{0,1,2,\ldots,M\}.
\]
We assume $y_i\in \{0,1,2,\ldots,M\}$ for some $M\ge 0$, since labels can always be shifted and discretized (as computers have finite precision). We further assume $M=\mathrm{poly}(m)$, so that it can be represented using $O(\log m)$ bits, which is typical in the RAM model.
 For a split $j \in \{0,1,\dots,N\}$, let
\[
L_j := [1,j], \qquad R_j := [j+1,N], \qquad
\mu(j) = \frac{\sum_{i=1}^m \mathbf{1}[x_i \le j]\, y_i}{\sum_{i=1}^m \mathbf{1}[x_i \le j]},
\qquad
\gamma(j) = \frac{\sum_{i=1}^m \mathbf{1}[x_i > j]\, y_i}{\sum_{i=1}^m \mathbf{1}[x_i > j]}.
\]
Note that $\mu(j)$ and $\gamma(j)$ are the label means on the left and right sides. They serve as the optimal prediction for $x_i \le j$ and $x_i \ge j+1$ if we split the data at $j$. The mean squared loss for regression is defined as follows
\begin{align} \label{eq:mse-loss}
L_{\mathrm{MSE}}(j)=\frac{1}{m}\left(
\sum_{i=1}^m \mathbf{1}[x_i \le j] (y_i-\mu(j))^2
+
\sum_{i=1}^m \mathbf{1}[x_i > j] (y_i-\gamma(j))^2
\right).
\end{align}

Our goal is to output $\widehat{j}$ such that $L_{\mathrm{MSE}}(\widehat{j})\le \opt+\epsilon$, where
\[
\opt:=\min_{0\le j\le N}L_{\mathrm{MSE}}(j).
\]

\begin{figure}[t]
\centering
\begin{tikzpicture}[scale=0.6]
  \draw[->,thick] (-0.2,0) -- (9.8,0) node[right] {$x$};
  \draw[->,thick] (0,-0.2) -- (0,4.8) node[above] {$y$};
  \node[below] at (0.3,0) {\small $1$};
  \node[below] at (9.2,0) {\small $N$};

  \foreach \x/\y in {0.4/1.0, 0.9/0.7, 1.4/1.3, 2.0/0.5, 2.5/1.1, 2.9/0.8, 3.4/1.4, 3.9/0.9, 4.3/1.2}{
    \draw[blue!50,thin] (\x,\y) -- (\x,1.0);
    \fill[blue!70!black] (\x,\y) circle (2.2pt);
  }

  \foreach \x/\y in {5.1/2.9, 5.6/3.4, 6.1/2.6, 6.5/3.2, 7.0/3.6, 7.5/2.8, 8.0/3.3, 8.5/3.0, 9.0/3.5}{
    \draw[red!50,thin] (\x,\y) -- (\x,3.1);
    \fill[red!70!black] (\x,\y) circle (2.2pt);
  }

  \draw[blue!70!black,thick] (0.2,1.0) -- (4.5,1.0)
    node[right,font=\small] {$\mu(j^\star)$};
  \draw[red!70!black,thick] (5.0,3.1) -- (9.3,3.1)
    node[right,font=\small] {$\gamma(j^\star)$};

  \draw[dashed,thick,gray!80] (4.7,-0.4) -- (4.7,4.5)
    node[above,black,font=\small] {$j^\star$};

  \draw[<->,blue!70!black] (0.2,-0.7) -- (4.5,-0.7)
    node[midway,below,font=\small] {$L_{j^\star}$};
  \draw[<->,red!70!black]  (5.0,-0.7) -- (9.2,-0.7)
    node[midway,below,font=\small] {$R_{j^\star}$};
\end{tikzpicture}
\caption{Illustration of the optimal split $j^\star$.}
\label{fig:optimal-split}
\end{figure}

\subparagraph{Classification.} For classification, the labels are in $\{-1,+1\}$. The misclassification loss at split $j$ is
\[
L_{\rm mis}(j)
=
\frac{1}{m}\Big(
\min\{f_{-1,[1,j]},f_{+1,[1,j]}\}
+
\min\{f_{-1,[j+1,N]},f_{+1,[j+1,N]}\}
\Big) \, .
\]
Here, for any interval $R \subseteq [N]$, $f_{+1,R}$ and $f_{-1,R}$ denote the numbers of data points in $R$ with labels $+1$ and $-1$, respectively. Intuitively, we classify each data point based on the majority on the side of the split. Another popular loss function is the Gini impurity (to be defined formally in Section \ref{sec:one-pass-gini}):
\[
L_{\mathrm{Gini}}(j)
=
\frac{|L_j|}{m}\operatorname{Gini}(\mset{ y_i : x_i \le j })
+
\frac{|R_j|}{m}\operatorname{Gini}(\mset{ y_i : x_i \ge j + 1 }).
\]

\subparagraph{Motivation.} Our objective is  to design optimal algorithms that  use memory sublinear in $m$ and $N$ and to prove matching lower bounds. 

In many large-scale applications, the data does not fit in main memory and must be processed as it arrives while using sublinear space. As more data arrive, $m$ grows and one will run out of main memory (RAM) if we choose to store all the data points. 

High-cardinality or composite features (i.e., a feature that is a combination
of several features) may also yield a large feature range $N$, significantly
impacting memory and time complexity~\cite{FayyadI92}.

Note that even if $O(N)$ space is acceptable, there might be many features and
we want to find the best split among all of them. Furthermore, in ensemble methods, we also often train many trees in parallel. Hence, using $o(N)$ space per feature and tree significantly improves the overall memory footprint.

Optimal split selection also arises outside tree learning, including in image segmentation and change-point detection \cite{otsu1975, aminikhanghahi2017}.

Providing matching lower and upper bounds is also of independent theoretical interest, as it characterizes the intrinsic difficulty of the problem and helps explain the memory requirements of practical heuristics.

\subparagraph{Our results.} In this paper, we focus on the one-pass regime. Our main algorithmic contribution is a one-pass additive approximation algorithm for regression using $\tO{M^2/\epsilon}$ space, which reduces to $\tO{1/\epsilon}$ when $M=\bigO{1}$. This improves upon the two-pass algorithm of \cite{PhamTV25}.

We further show that the same framework yields a one-pass additive approximation algorithm for Gini impurity using $\tO{1/\epsilon}$ space, improving upon the previous $\tO{1/\epsilon^2}$ upper bound in \cite{PhamTV25}. 

On the lower-bound side, we prove one-pass space lower bounds for regression, misclassification, and Gini impurity. In particular, our upper and lower bounds are essentially tight, up to polylogarithmic factors, for regression, misclassification, and Gini impurity. Table~\ref{tab:upper_comparison} summarizes the relevant bounds for numerical split objectives.

Our algorithms succeed with high probability, for example, with probability at least $1 - 1/\mathrm{poly}(N)$ or $1 - 1/\mathrm{poly}(m)$. Our lower bounds apply to any algorithm that succeeds with probability at least $2/3$. The deterministic complexity of the problem remains an interesting open question.

All omitted proofs can be found in  Appendix \ref{sec:omitted}.

\begin{table}[t]
\centering
\small
\renewcommand{\arraystretch}{1.25}
\begin{tabular}{|l|c|l|l|}
\hline
\textbf{Objective} & \textbf{Passes} & \textbf{Space} & \textbf{Reference} \\
\hline
\multicolumn{4}{|l|}{\textit{Upper bounds}} \\
\hline
Regression  & $2$ & $\tO{M^2/\epsilon}$ & \cite{PhamTV25} \\
\hline
\rowcolor{gray!10}
Regression  & $1$ & $\tO{M^2/\epsilon}$ & This Paper \\
\hline
Misclassification & $1$ & $\tO{1/\epsilon}$ & \cite{PhamTV25} \\
\hline
Gini impurity & $1$ & $\tO{1/\epsilon^2}$ & \cite{PhamTV25} \\
\hline
\rowcolor{gray!10} Gini impurity & $1$ & $\tO{1/\epsilon}$ & This Paper \\
\hline
\multicolumn{4}{|l|}{\textit{Lower bounds}} \\
\hline
\rowcolor{gray!10}
Regression & $1$ & $\Omega(M^2/\epsilon)$ & This Paper \\
\hline
\rowcolor{gray!10}
Misclassification & $1$ & $\Omega(1/\epsilon)$ & This Paper \\
\hline
\rowcolor{gray!10}
Gini impurity & $1$ & $\Omega(1/\epsilon)$ & This Paper \\
\hline
\end{tabular}
\caption{Comparison of upper and lower space bounds for numerical split objectives in the streaming model. All guarantees are additive-$\epsilon$ approximations.}
\label{tab:upper_comparison}
\end{table}


\section{Algorithms}

\subsection{One-pass additive approximation for regression}
\label{sec:one-pass-regression}

In this section we present a truly one-pass additive approximation algorithm for the regression split problem. In contrast with the previous algorithm in \cite{PhamTV25}, the algorithm does not require the stream length $m$ to be known in advance. 

At a high level, we obtain a set of candidate splits via reservoir sampling. We show that the loss function satisfies a Lipschitz property which implies that one of the candidates is a good approximation. To this end, the loss of each candidate split is estimated using a coupled dyadic Count--Min sketch that tracks the three range moments required for regression.

The main goal of this section is to prove the following.

\begin{theorem}
\label{thm:one-pass-regression}
Fix $\epsilon\in(0,1)$. There exists a randomized one-pass streaming algorithm for the regression split problem that, on every insertion-only stream with labels in $\{0,1,\ldots,M\}$, uses $\widetilde{\mathcal{O}}\!\left(\frac{M^2}{\epsilon}\right)$ space and post-processing time, has $\tO{1}$ update time, and with high probability outputs a split $\widehat{j}\in\{0,1,\ldots,N\}$ satisfying $L_{\mathrm{MSE}}(\widehat{j})\le \OPT+\epsilon$.
\end{theorem}

We first need to define several quantities to be used throughout this section. For every range $R \subseteq [1,N]$, let the count, 1st moment, and 2nd moment be defined as
\[
  n_R :=\sum_{i=1}^m \mathbf{1}[x_i\in R], \qquad
  s_R :=\sum_{i=1}^m \mathbf{1}[x_i\in R]\,y_i, \qquad
  q_R :=\sum_{i=1}^m \mathbf{1}[x_i\in R]\,y_i^2.
\]

For every nonempty range $R$, let $\bar y_R := \frac{s_R}{n_R}$. Observe that
\[
\sum_{i:x_i\in R}(y_i-\bar y_R)^2
= q_R - 2\bar y_R s_R + n_R\bar y_R^2
= q_R - \frac{2s_R^2}{n_R} + \frac{s_R^2}{n_R}
= q_R - \frac{s_R^2}{n_R}.
\]
We define the sum squared error of a set as follows.
\[
\SSE(R):=
\begin{cases}
q_R-\dfrac{s_R^2}{n_R}, & n_R>0,\\
0, & n_R=0.
\end{cases}
\]

For every split $j$, define $L_j := \{1,2,\ldots,j\}$ and $R_j := \{j+1,j+2,\ldots,N\}$. We have
\begin{equation}
m L_{\mathrm{MSE}}(j)=\SSE(L_j)+\SSE(R_j).
\label{eq:sse-identity}
\end{equation}

Throughout this section, we will use the following quantities:
\[
\tau:=\frac{\epsilon}{16M^2},
\qquad
\beta:=\frac{\epsilon}{32M^2},
\qquad
K:=\Theta\!\left(\frac{\log N}{\tau}\right).
\]


We rely on the classic Count--Min sketch with range queries due to Cormode and Muthukrishnan~\cite{CormodeM05}.

\begin{theorem}[Count--Min sketch with range queries]
\label{thm:cms-range}
Consider an insertion-only stream of tuples
$  (x_1,y_1), (x_2,y_2), \ldots, (x_m,y_m)$, where $x_i \in [N]$ and $y_i \in \{0,1,\ldots,M\} $. For any interval $R = [l,r] \subseteq [N]$,
define
\[
  f_R \coloneqq \sum_{i :\, x_i \in R} y_i,
  \qquad\text{and}\qquad
  W \coloneqq \sum_{i=1}^m y_i.
\]
Then there exists a data structure using $\tO{k}$ bits of space and $\tO{1}$ update time such that, with high probability, for all intervals $R = [l,r] \subseteq [N]$, it returns an estimate $\widehat{f}_R$ satisfying
\[
  f_R \le \widehat{f}_R \le f_R + \frac{W}{k}.
\]
\end{theorem}

We outline our algorithm below.

\begin{tcolorbox}[
  title=Algorithm for estimating regression split,
  fonttitle=\bfseries, breakable]

\begin{enumerate}
  \item During the stream, maintain a reservoir sample of size $K$, and after the stream define
\[
S:=\{0,N\}\cup \{x,\ x-1 : (x,y)\text{ appears in the reservoir and }x>1\}.
\]

  \item Simultaneously, maintain three Count--Min range-query sketches with $k=1/\beta$ for the streams
\[
  \{ (x_i, 1) \}_{i\in [m]}, \qquad \{ (x_i, y_i) \}_{i\in [m]}, \qquad \{ (x_i, y_i^2) \}_{i\in [m]}.
\]
That is when $(x_i,y_i)$ arrives, we insert $(x_i,1)$ into the first sketch, $(x_i,y_i)$ into the second, and $(x_i,y_i^2)$ into the third.

  \item For any interval $R\subseteq[N]$, querying the sketches returns estimates $\widehat n_R,\widehat s_R,\widehat q_R$. Set
\[
\widehat{\SSE}(R):=
\begin{cases}
\widehat q_R-\dfrac{\widehat s_R^2}{\widehat n_R + \beta m}, & \widehat n_R>0,\\
0, & \widehat n_R=0.
\end{cases}
\]

  \item For every candidate split $j\in S$, define
\[
\widehat L_{\mathrm{MSE}}(j):=\frac{1}{m}\Bigl(\widehat{\SSE}([1,j])+\widehat{\SSE}([j+1,N])\Bigr).
\]

  \item Output $\widehat j:=\arg\min_{j\in S}\widehat L_{\mathrm{MSE}}(j).$
\end{enumerate}
\end{tcolorbox}

\begin{lemma}
\label{lem:cms-moments-independent}
With high probability, the following holds simultaneously for all intervals \(R\subseteq[N]\):
\[
\widehat n_R=n_R+a_R,\qquad
\widehat s_R=s_R+b_R,\qquad
\widehat q_R=q_R+c_R,
\]
where
\[
0\le a_R\le \beta m,\qquad
0\le b_R\le \beta M m,\qquad
0\le c_R\le \beta M^2 m.
\]
\end{lemma}

\begin{proof}
Apply Theorem~\ref{thm:cms-range} with \(k=\Theta(1/\beta)\) separately to the three insertion-only streams
\[
\{(x_i,1)\}_{i=1}^m,\qquad
\{(x_i,y_i)\}_{i=1}^m,\qquad
\{(x_i,y_i^2)\}_{i=1}^m.
\]
For the first stream, the total weight is \(m\). For the second stream, the total weight is at most \(mM\). Finally, for the third stream, the total weight is at most \(mM^2\). The claim follows.
\end{proof}

\begin{lemma}
\label{lem:loss-approximation}
Condition on the event in Lemma~\ref{lem:cms-moments-independent}. Then for every candidate split \(j\in S\),
\[
|\widehat L_{\mathrm{MSE}}(j)-L_{\mathrm{MSE}}(j)|\le \frac{\epsilon}{4}.
\]
\end{lemma}

\begin{proof}
Fix $j\in S$ and let $R$ be either $L_j$ or $R_j$. Set $D:=\widehat n_R+\beta m=n_R+a_R+\beta m$, so $D\ge n_R$ and $D\ge\beta m$.

If $n_R=0$, then $s_R=q_R=0$ and $\SSE(R)=0$. If $\widehat n_R=0$, then $\widehat{\SSE}(R)=0=\SSE(R)$, so the claim is immediate. Otherwise, if $\widehat n_R>0$, then since $s_R=q_R=0$, Lemma~\ref{lem:cms-moments-independent} gives $\widehat q_R=q_R+c_R=c_R$ and $\widehat s_R=s_R+b_R=b_R$; substituting into $\widehat{\SSE}(R)=\widehat q_R-\widehat s_R^2/D$ and using $\SSE(R)=0$ yields
\[
\widehat{\SSE}(R)-\SSE(R)=c_R-\frac{b_R^2}{D}.
\]
Since $c_R\ge 0$ we get $\widehat{\SSE}(R)-\SSE(R)\le c_R\le\beta M^2m$, and since $D\ge\beta m$ we get
\[
\widehat{\SSE}(R)-\SSE(R)\ge -\frac{b_R^2}{D}\ge -\frac{(\beta Mm)^2}{\beta m}=-\beta M^2m.
\]
Hence $|\widehat{\SSE}(R)-\SSE(R)|\le\beta M^2m$.

Henceforth assume $n_R\ge 1$. Observe that
\[
\widehat{\SSE}(R)-\SSE(R)
= \Bigl(q_R+c_R - \frac{(s_R+b_R)^2}{D}\Bigr) - \Bigl(q_R - \frac{s_R^2}{n_R}\Bigr)
= c_R + \frac{s_R^2}{n_R} - \frac{(s_R+b_R)^2}{D}.
\]
We can upper bound the difference as follows:
\begin{align*}
\widehat{\SSE}(R)-\SSE(R)
&= c_R + \frac{s_R^2}{n_R} - \frac{(s_R+b_R)^2}{D} \\
&\le c_R + \frac{s_R^2}{n_R} - \frac{s_R^2}{D}
  &&\text{(drop $-2s_Rb_R/D - b_R^2/D \le 0$)}\\
&= c_R + \frac{s_R^2(D-n_R)}{n_RD} \\
&\le \beta M^2m + \frac{M^2n_R^2\cdot 2\beta m}{n_R D}
  &&\text{($s_R\le Mn_R$, and $D-n_R\le 2\beta m$)}\\
&\le 3\beta M^2m.
  &&\text{($n_R/D \le 1$)}
\end{align*}
We now lower bound the difference:
\begin{align*}
\widehat{\SSE}(R)-\SSE(R)
&= c_R + \frac{s_R^2}{n_R} - \frac{(s_R+b_R)^2}{D} \\
&\ge - \frac{2s_Rb_R}{D} - \frac{b_R^2}{D}
  &&\text{(drop $c_R\ge 0$ and $s_R^2/n_R\ge s_R^2/D\ge 0$)}\\
&\ge -\frac{2Mn_R\cdot\beta Mm}{n_R} - \frac{(\beta Mm)^2}{\beta m}
  &&\text{($D\ge n_R$; and $D\ge\beta m$)}\\
&= -2\beta M^2m - \beta M^2m = -3\beta M^2m.
\end{align*}
Hence $|\widehat{\SSE}(R)-\SSE(R)|\le 3\beta M^2m$. Applying this to both sides and using $mL_{\mathrm{MSE}}(j)=\SSE(L_j)+\SSE(R_j)$,
\[
|\widehat L_{\mathrm{MSE}}(j)-L_{\mathrm{MSE}}(j)|\le \frac{6\beta M^2m}{m} = 6\beta M^2 = \frac{6M^2\epsilon}{32M^2} = \frac{3\epsilon}{16} < \frac{\epsilon}{4}.\qedhere
\]
\end{proof}

The next lemma is from Pham et al. \cite{PhamTV25}. We include the proof, simplified and with the constant \(4\) removed, in the appendix for completeness. It bounds the difference between the squared loss at splits $j'$ and $j$ based on the number of data points in $(j, j']$. 

\begin{lemma}\label{lem:split-shift}
Let $0\le j<j'\le N$, and let $b:=\bigl|\{i:\ j<x_i\le j'\}\bigr|$.
Then
\[
|L_{\mathrm{MSE}}(j')-L_{\mathrm{MSE}}(j)|\le \frac{bM^2}{m}.
\]
\end{lemma}

The next lemma is also adapted from \cite{PhamTV25}. Roughly speaking, if we sample $\approx C M^2/\epsilon \cdot \log N $ data points for some sufficiently large constant $C$, then the probability that an interval with at least $\epsilon  m/M^2$ data points has no data point in the sample is upper bounded by 
\[
  \approx \left( 1- \frac{C M^2 \log N}{m \epsilon} \right)^{\epsilon  m/M^2} \le e^{- \frac{\epsilon  m}{M^2} \cdot \frac{C M^2 \log N}{m \epsilon }} \le \frac{1}{N^4}.
\]
Hence, taking a union bound over $\binom{N}{2}$ intervals, we deduce that this could not happen with high probability. This implies that one of the candidate splits must be a good approximation.

\begin{lemma}
\label{lem:heavy-hit}
With probability at least $1-1/N^2$, every interval $(a,b]\subseteq[N]$ containing more than
$\tau m$ stream items contains at least one sampled item from the final reservoir of size $K=\lceil C\log N/\tau\rceil$, for a sufficiently large constant $C$.
\end{lemma}

\begin{lemma}
\label{lem:candidate-quality}
With probability at least $1-1/N^2$, there exists a candidate $j\in S$ such that
\[
L_{\mathrm{MSE}}(j)\le \OPT+\frac{\epsilon}{16}.
\]
\end{lemma}

\begin{proof}
Let $j^\star$ be an optimal split, so $L_{\mathrm{MSE}}(j^\star)=\OPT$. Condition on the event in
Lemma~\ref{lem:heavy-hit}.

If $j^\star\in S$, then we are done. Otherwise, since $N\in S$, the set $\{t\in S:t>j^\star\}$ is nonempty. Let
\[
j:=\min\{t\in S:t>j^\star\}.
\]
We claim that the interval $(j^\star,j]$ contains no sampled feature value.

Indeed, if some sampled point had feature value $x\in(j^\star,j)$, then $x\in S$, contradicting
the minimality of $j$. If some sampled point had feature value $x=j$, then $x-1\in S$ by the
definition of $S$, and since $x-1\ge j^\star$, this again contradicts the minimality of $j$
unless $x-1=j^\star$, in which case $j^\star\in S$, contradiction.

Hence $(j^\star,j]$ contains no sampled point. By the event in Lemma~\ref{lem:heavy-hit}, this
interval must therefore contain at most $\tau m$ stream items. Applying
Lemma~\ref{lem:split-shift},
\[
L_{\mathrm{MSE}}(j)-L_{\mathrm{MSE}}(j^\star)\le \frac{\tau m\cdot M^2}{m}=\tau M^2=\frac{\epsilon}{16}.
\]
Since $L_{\mathrm{MSE}}(j^\star)=\OPT$, we conclude that
\[
L_{\mathrm{MSE}}(j)\le \OPT+\frac{\epsilon}{16}. \qedhere
\]
\end{proof}

We now put it all together to prove Theorem~\ref{thm:one-pass-regression}.

\begin{proof}[Proof of Theorem~\ref{thm:one-pass-regression}]
If $M^2\le \epsilon/16$ then we output $\widehat{j} = 0$; the average squared loss is at most $M^2 \le \epsilon/16 < \epsilon$. Henceforth assume $M^2 > \epsilon/16$.

The algorithm maintains a reservoir sample of size $K$ and range-query Count--Min sketches with parameter $\beta=\Theta(\epsilon/M^2)$. The reservoir has $O(1)$ update time and the sketch has $\tO{1}$ update time.

By Lemma~\ref{lem:candidate-quality}, with probability at least $1-1/N^2$, there exists
$j\in S$ such that
\[
L_{\mathrm{MSE}}(j)\le \OPT+\frac{\epsilon}{16}.
\]
By Lemma~\ref{lem:loss-approximation}, with probability at least $1-1/N^2$, every
candidate satisfies $|\widehat L_{\mathrm{MSE}}(j)-L_{\mathrm{MSE}}(j)|\le \epsilon/4$.
Hence both events hold simultaneously with high probability. Since $\widehat j$ minimizes $\widehat L_{\mathrm{MSE}}$ over $S$,
\[
L_{\mathrm{MSE}}(\widehat j)
\le \widehat L_{\mathrm{MSE}}(\widehat j)+\frac{\epsilon}{4}
\le \widehat L_{\mathrm{MSE}}(j)+\frac{\epsilon}{4}
\le L_{\mathrm{MSE}}(j)+\frac{\epsilon}{2}
\le \OPT+\frac{\epsilon}{16}+\frac{\epsilon}{2}
< \OPT+\epsilon. \qedhere
\]
\end{proof}

\subsection{One-pass additive approximation for classification with Gini impurity}\label{sec:one-pass-gini}

In this section, we improve the $\tO{1/\epsilon^2}$-space algorithm of \cite{PhamTV25} to $\tO{1/\epsilon}$ to match our lower bound in the next section.  For a split $j$, write
\[
a=f_{+1,[1,j]},\qquad b=f_{-1,[1,j]},\qquad c=f_{+1,[j+1,N]},\qquad d=f_{-1,[j+1,N]}.
\]
The Gini impurity of a multiset $S$ of labels is $1 - \sum_{y} \left(\frac{|S_y|}{|S|}\right)^2$, where $S_y = \mset{s \in S : s = y}$. We adopt the convention $\operatorname{Gini}(\emptyset)=0$. For example, if all labels are identical then the impurity is $1 - 1^2 = 0$ (a pure node), whereas if the labels are split evenly between $+1$ and $-1$ the impurity is $1 - \left(\tfrac{1}{2}\right)^2 - \left(\tfrac{1}{2}\right)^2 = \tfrac{1}{2}$, which is the maximum.

For binary labels this equals $\frac{2pq}{(p+q)^2}$ where $p = |S_{+1}|$ and $q = |S_{-1}|$.
The Gini loss of split $j$ is the weighted sum of impurities of the two children:
\[
L_{\mathrm{Gini}}(j) = \frac{a+b}{m}\cdot\frac{2ab}{(a+b)^2} + \frac{c+d}{m}\cdot\frac{2cd}{(c+d)^2}=
\frac{2ab}{m(a+b)}+\frac{2cd}{m(c+d)}.
\]

\begin{lemma}\label{lem:gini-shift}
For any splits $j < j'$ in $[N]$, let $\ell = \bigl|\{ i : j < x_i \le j' \}\bigr|$ be the number of data points in the interval $(j, j']$. Then
\[
\bigl|L_{\text{Gini}}(j) - L_{\text{Gini}}(j')\bigr| \le \frac{2\ell}{m}.
\]
\end{lemma}
\begin{proof}[Proof sketch]
We re-encode each label $y_i\in\{-1,+1\}$ as
$z_i=(y_i+1)/2\in\{0,1\}$ and establish the identity $2L_{\mathrm{MSE}}(j)=L_{\mathrm{Gini}}(j)$. 
The Lipschitz bound then follows by applying Lemma~\ref{lem:split-shift} with $M=1$.
\end{proof}

\begin{theorem}\label{thm:one-pass-gini}
Fix $\epsilon\in(0,1)$. There exists a randomized one-pass streaming algorithm for the Gini split problem that uses $\tO{1/\epsilon}$ space, has $\tO{1}$ update time, and with high probability outputs $\widehat{j}\in\{0,\ldots,N\}$ satisfying $L_{\mathrm{Gini}}(\widehat{j})\le\OPT_{\mathrm{Gini}}+\epsilon$.
\end{theorem}

\begin{proof}
Re-encode each label $y_i\in\{-1,+1\}$ as $z_i=(y_i+1)/2\in\{0,1\}$. As shown in the proof of Lemma~\ref{lem:gini-shift}, for every split $j$,
\[
L_{\mathrm{Gini}}(j) = 2\,L_{\mathrm{MSE}}(j),
\]
where $L_{\mathrm{MSE}}$ is the standard mean squared loss on the re-encoded stream. In particular, $\OPT_{\mathrm{Gini}}=2\,\OPT_{\mathrm{MSE}}$ and both objectives share the same optimal split.

Apply the algorithm of Theorem~\ref{thm:one-pass-regression} to the stream $(x_i,z_i)$ with $M=1$ and target error $\epsilon/2$. It uses $\tO{M^2/(\epsilon/2)}=\tO{1/\epsilon}$ space and, with high probability, outputs $\widehat{j}$ satisfying
\[
L_{\mathrm{MSE}}(\widehat{j})\le\OPT_{\mathrm{MSE}}+\frac{\epsilon}{2}.
\]
Multiplying through by $2$:
\[
L_{\mathrm{Gini}}(\widehat{j})=2\,L_{\mathrm{MSE}}(\widehat{j})\le 2\,\OPT_{\mathrm{MSE}}+\epsilon=\OPT_{\mathrm{Gini}}+\epsilon.\qedhere
\]
\end{proof}

\section{Lower bounds}\label{sec:lower_bound}

In this section, we establish one-pass space lower bounds for approximating regression and classification optimal splits. The lower bounds for numerical observations all follow the same high-level template:
we reduce from \textsc{Index} by constructing a stream with two competing adjacent candidate splits whose identity reveals the hidden bit.

\begin{lemma}[\textsc{Index} one-way lower bound~\cite{Ablayev96}]\label{lem:index}
In the one-way randomized communication problem \textsc{Index},
Alice holds $z\in\{0,1\}^n$ and Bob holds $i\in[n]$; Bob must output $z_i$.
Any one-way randomized protocol that succeeds with probability at least $2/3$
requires $\Omega(n)$ bits of communication.
\end{lemma}

The main difference between the proofs lies in the loss-specific gap calculation:
quadratic for regression, direct majority/minority counting for misclassification.

\subsection{Lower bound for regression}\label{subsec:lb-regression}

Fix $\epsilon\in(0,10^{-3})$. Consider the problem of finding the optimal regression split with labels $y_i\in[0,1]$ (i.e., $M=1$)
and feature values $x_i\in[N]$. Our goal is to show that any one-pass randomized streaming algorithm that, with probability at least $2/3$, outputs a split $\widehat{j}\in[N]$ satisfying
$L_{\mathrm{MSE}}(\widehat{j})\le \opt+\epsilon$ must use $\Omega(1/\epsilon)$ bits of memory, even for instances with $N=\Theta(1/\epsilon)$.

The following lemma is well-known. It shows that given a multiset $S$ of numbers, then $\sum_{x \in S} (x - d)^2$ is minimized when $d$ is the average of elements in $S$.

\begin{lemma}\label{lem:mean-min}
Let $y_1,\dots,y_k\in\mathbb{R}$ and let $\bar y := \frac1k\sum_{t=1}^k y_t$.
Then for any $a\in\mathbb{R}$,
\[
\sum_{t=1}^k (y_t-a)^2
=
\sum_{t=1}^k (y_t-\bar y)^2
\;+\;k(a-\bar y)^2.
\]
In particular, $\sum_{t=1}^k (y_t-\bar y)^2 \le \sum_{t=1}^k (y_t-a)^2$ for all $a$.
\end{lemma}

The next lemma quantifies how much the squared error objective changes when we append a uniform block of labels to an existing set.

\begin{lemma}\label{lem:add-block}
Let $D$ be a size-$a$ multiset of reals with mean $\mu$, and define
$\SSE(D):=\sum_{y\in D}(y-\mu)^2$.
Let $G$ be a multiset of size $B$ whose elements are all equal to the same value $v$
(so $\SSE(G)=0$). Then
\[
\SSE(D\cup G)=\SSE(D)\;+\;\frac{aB}{a+B}\,(\mu-v)^2.
\]
\end{lemma}

\begin{theorem}\label{thm:lb-regression}
For all sufficiently small $\epsilon>0$, the following holds.
Consider the optimal regression split problem with labels $y_t\in[0,M]$
and feature values $x_t\in[N]$.
Any one-pass randomized streaming algorithm that, with probability at least $2/3$,
outputs a split $\widehat j\in[N]$ satisfying
\[
L_{\mathrm{MSE}}(\widehat j)\le \opt+\epsilon
\]
must use $\Omega(M^2/\epsilon)$ bits of memory.
\end{theorem}

\begin{proof}
We first remove $M$ by scaling.
Given labels $y_t\in[0,M]$, define $\tilde y_t := y_t/M\in[0,1]$.
For every split $j$,
\[
L_{\mathrm{MSE}}(j;\, y)=M^2\, L_{\mathrm{MSE}}(j;\, \tilde y).
\]
Therefore, an $\epsilon$-additive algorithm for labels in $[0,M]$
would give an $(\epsilon/M^2)$-additive algorithm for labels in $[0,1]$.
So it suffices to prove the theorem for $M=1$ with target additive error:
\[
\epsilon' := \epsilon/M^2.
\]
We may assume $\epsilon'\le 10^{-3}$.

We reduce from \textsc{Index}. Alice encodes her input as a data stream and runs the algorithm on it, then sends the algorithm's memory state to Bob. Bob appends his part of the stream and uses the algorithm's output to determine whether the queried bit is $0$ or $1$. Let
\[
n:=\left\lfloor \frac{1}{1000\epsilon'}\right\rfloor,\qquad
B:=n,\qquad
T:=100n^2,\qquad
N:=2n+1.
\]
Alice and Bob will build a stream over labels in $\{0,1\}$.

\textbf{Alice's prefix.}
Given $z\in\{0,1\}^n$, for each $k\in[n]$, Alice inserts $B$ copies of
\[
(x,y)=(2k,z_k).
\]
Thus Alice contributes $m_A=nB=n^2$ points.

\textbf{Bob's suffix.}
Given $i\in[n]$, Bob appends
\[
T\text{ copies of }(2i-1,0)
\qquad\text{and}\qquad
T\text{ copies of }(2i+1,1).
\]
The total stream length is
\[
m=m_A+2T=n^2+200n^2=201n^2.
\]

\begin{figure}[h]
\centering
\begin{tikzpicture}[scale=1.2]
  \draw[->,thick] (0.5,0) -- (9.8,0) node[right] {$x$};
  \foreach \x in {1,...,9}{
    \draw (\x,0.06) -- (\x,-0.06);
    \node[below=3pt,font=\small] at (\x,0) {$\x$};
  }

  \foreach \x/\lbl in {2/0, 4/1, 6/1, 8/0}{
    \node[draw=blue!70!black, fill=blue!10, minimum width=0.55cm, minimum height=0.38cm,
          font=\small, text=blue!70!black] at (\x, 0.55) {$B$};
    \node[above=2pt, font=\small, blue!70!black] at (\x, 0.74) {$\lbl$};
  }
  \node[left,font=\footnotesize,blue!70!black] at (0.7,0.55) {Alice};

  \node[draw=red!70!black, fill=red!10, minimum width=0.55cm, minimum height=0.38cm,
        font=\small, text=red!70!black] at (3,-0.65) {$T$};
  \node[below=2pt, font=\small, red!70!black] at (3,-0.74) {$0$};

  \node[draw=red!70!black, fill=red!10, minimum width=0.55cm, minimum height=0.38cm,
        font=\small, text=red!70!black] at (5,-0.65) {$T$};
  \node[below=2pt, font=\small, red!70!black] at (5,-0.74) {$1$};

  \node[left,font=\footnotesize,red!70!black] at (0.7,-0.55) {Bob};

  \draw[dashed,gray,thick] (3,1) -- (3,1.3)
    node[above,font=\small,black] {$j^-$};
  \draw[dashed,orange!80!black,thick] (4,1) -- (4,1.3)
    node[above,font=\small,black] {$j^+$};
\end{tikzpicture}
\caption{Construction for $n=4$, $z=(0,1,1,0)$, and $i=2$.
Blue (above): Alice's $B$ copies of each of $(2k,z_k)$.
Red (below): Bob's $T$ copies of $(3,0)$ and $(5,1)$.
Dashed lines mark the two candidate splits $j^-=3$ and $j^+=4$.}
\label{fig:lb-construction}
\end{figure}

Consider the following two splits:
\[
j^-:=2i-1,\qquad j^+:=2i.
\]
These are the two adjacent splits around the block at $x=2i$.

We first claim that any split that is not $j^-$ or $j^+$ has high squared error.

\begin{claim}
For every $j\notin\{j^-,j^+\}$, we have
\[
L_{\mathrm{MSE}}(j)\ge \frac{50}{201}.
\]
\end{claim}

\begin{proof}
Fix $j\notin\{j^-,j^+\}$. If $j\le 2i-2$, then both anchor blocks, namely the $T$ copies of $(2i-1,0)$
and the $T$ copies of $(2i+1,1)$, lie on the right side of the split.
Let $c$ be the optimal constant on the right side, i.e., the average label on that side.
Then the contribution of the anchor blocks alone is
\[
T(0-c)^2+T(1-c)^2.
\]
This is minimized at $c=\tfrac12$, so
\[
T(0-c)^2+T(1-c)^2 \ge T\Bigl(0-\tfrac12\Bigr)^2 + T\Bigl(1-\tfrac12\Bigr)^2 = \frac{T}{2}.
\]

If $j\ge 2i+1$, then both anchor blocks lie on the left side of the split,
and the same argument shows that the anchor contribution on the left is at least $T/2$.

Thus, for every $j\notin\{j^-,j^+\}$, the total squared error is at least $T/2$.
Since $m=201n^2$ and $T=100n^2$, we obtain
\[
L_{\mathrm{MSE}}(j)\ge \frac{T/2}{m} = \frac{100n^2/2}{201n^2} = \frac{50}{201}. \qedhere
\]
\end{proof}

We now show that both $j^-$ and $j^+$ have much lower squared error. Hence, the optimal split must be among them.

\begin{claim}
For $j\in\{j^-,j^+\}$, we have $L_{\mathrm{MSE}}(j)\le \frac{1}{201}$.
\end{claim}
\begin{proof}
Recall that $j^-$ splits the feature values into intervals $[1,2i-1]$ and $[2i,N]$; on the other hand, $j^+$ splits the feature values into intervals $[1,2i]$ and $[2i+1,N]$.

In both cases, all $T$ zeros at $x=2i-1$ lie on the left of $j$ and all $T$ ones at $x=2i+1$ lie on the right.
Using constant $0$ on the left and constant $1$ on the right, both anchor blocks incur zero error.
Only Alice's points contribute, and each contributes at most $1$, so
\[
L_{\mathrm{MSE}}(j)\le \frac{m_A}{m} = \frac{n^2}{201n^2} = \frac1{201}.\qedhere
\]
\end{proof}
Since $\frac1{201}<\frac{50}{201}$, we must have
\[
\opt=\min\{L_{\mathrm{MSE}}(j^-),L_{\mathrm{MSE}}(j^+)\}.
\]

{\bf The better split between $j^-$ and $j^+$ reveals $z_i$.}
Let $v:=z_i\in\{0,1\}$.
Write
\[
D_L:=\{t:x_t<2i\},\qquad D_R:=\{t:x_t>2i\},
\]
and let
\[
a:=|D_L|,\quad \mu:=\text{mean label on }D_L,
\qquad
b:=|D_R|,\quad \gamma:=\text{mean label on }D_R.
\]

Let $G$ be the block of $B$ copies of label $v$ at $x=2i$. Under $j^+$, the block $G$ is placed on the left, and under $j^-$, it is placed on the right.

Because $D_L$ contains exactly $T=100n^2$ zeros from Bob and at most $n^2$ ones from Alice,
\[
\mu\le \frac{n^2}{100n^2+n^2}=\frac{1}{101}<\frac1{100}.
\]
Similarly, $D_R$ contains exactly $100n^2$ ones from Bob and at most $n^2$ zeros from Alice, so
\[
\gamma\ge \frac{100n^2}{100n^2+n^2}=\frac{100}{101}>\frac{99}{100}.
\]

We note that the function $\frac{a}{a+B}$ is increasing with respect to $a$. Since $a,b\ge T=100n^2$ and $B=n$, we have
\[
\frac{a}{a+B}\ge \frac{T}{T+B} =\frac{100n^2}{100n^2+n} >\frac12,
\]

Similarly,
\[
\frac{b}{b+B}>\frac12.
\]

Under the split $j^+$, the block $G$ is appended to the left side $D_L$; under the split $j^-$, it is appended to the right side $D_R$.
Applying Lemma~\ref{lem:add-block} to each case,
\begin{align*}
L_{\mathrm{MSE}}(j^+) &= \frac{1}{m}\left[\SSE(D_L)+\frac{aB}{a+B}(\mu-v)^2 + \SSE(D_R)\right],\\
L_{\mathrm{MSE}}(j^-) &= \frac{1}{m}\left[\SSE(D_L) + \SSE(D_R)+\frac{bB}{b+B}(\gamma-v)^2\right].
\end{align*}
This gives
\[
L_{\mathrm{MSE}}(j^+)-L_{\mathrm{MSE}}(j^-)
=
\frac{B}{m}\left[
\frac{a}{a+B}(\mu-v)^2
-
\frac{b}{b+B}(\gamma-v)^2
\right].
\]

{\bf Case 1:} If $v=1$, then
  \[
    (\mu-1)^2\ge (99/100)^2,
    \qquad
    (\gamma-1)^2\le (1/100)^2.
  \]
  Hence,
  \[
    L_{\mathrm{MSE}}(j^+)-L_{\mathrm{MSE}}(j^-) 
    \ge 
      \frac{B}{m}\left[\frac12\Bigl(\frac{99}{100}\Bigr)^2-\Bigl(\frac1{100}\Bigr)^2\right]
    > 
      \frac{B}{4m} 
    = 
      \frac{1}{804\,n} \implies L_{\mathrm{MSE}}(j^-)<L_{\mathrm{MSE}}(j^+).
  \]

{\bf Case 2:} If $v=0$, then
  \[
    \mu^2\le (1/100)^2,
    \qquad
    \gamma^2\ge (99/100)^2.
  \]
  Hence,
  \[
    L_{\mathrm{MSE}}(j^+)-L_{\mathrm{MSE}}(j^-)
    \le
    \frac{B}{m}\left[\Bigl(\frac1{100}\Bigr)^2-\frac12\Bigl(\frac{99}{100}\Bigr)^2\right]
    <
    -\frac{B}{4m}
    =
    -\frac{1}{804\,n} \implies L_{\mathrm{MSE}}(j^+)<L_{\mathrm{MSE}}(j^-).
  \]

In both cases,
\[
|L_{\mathrm{MSE}}(j^+)-L_{\mathrm{MSE}}(j^-)|>\frac{1}{804\,n}.
\]
Since
\[
n=\left\lfloor \frac{1}{1000\epsilon'}\right\rfloor,
\qquad\text{hence}\qquad
\frac1n\ge 1000\epsilon',
\]
we get
\[
|L_{\mathrm{MSE}}(j^+)-L_{\mathrm{MSE}}(j^-)|>\epsilon'.
\]
Therefore any algorithm that outputs $\widehat j$ with
\[
L_{\mathrm{MSE}}(\widehat j)\le \opt+\epsilon'
\]
must return the correct minimizer in $\{j^-,j^+\}$, and hence reveals $z_i$.

So Alice can run the streaming algorithm on her prefix, send its memory state to Bob,
and Bob can recover $z_i$ with probability at least $2/3$.
This gives a one-way protocol for \textsc{Index} with communication equal to the memory used.
By Lemma~\ref{lem:index}, that memory must be $\Omega(n)$ bits. Finally,
\[
n=\Theta(1/\epsilon')=\Theta(M^2/\epsilon),
\]
so the space lower bound is $\Omega(M^2/\epsilon)$.
\end{proof}

\subsection{Lower bound for classification}\label{subsec:lb-classification}

For classification with numerical observations, we will show the following.
Fix $\epsilon\in(0,10^{-3})$.
Any one-pass randomized streaming algorithm that, with probability at least $2/3$,
outputs a split $\widehat{j}\in[N]$ satisfying
$L_{\rm mis}(\widehat{j})\le \opt+\epsilon$
must use $\Omega(1/\epsilon)$ bits of memory, even for instances with $N=\Theta(1/\epsilon)$.

For an interval $R\subseteq[N]$, let
\[
  f_{+1,R} = |\{ i : x_i \in R,\; y_i = +1 \}|, \qquad
  f_{-1,R} = |\{ i : x_i \in R,\; y_i = -1 \}|
\]
denote the number of $+1$ and $-1$ labels falling in $R$, respectively.

\begin{theorem}\label{thm:lb-misclassification}
Fix $\epsilon\in(0,10^{-3})$.
Any one-pass randomized streaming algorithm that, with probability at least $2/3$,
outputs a split $\widehat{j}\in[N]$ satisfying
\[
L_{\rm mis}(\widehat{j})\le \opt+\epsilon
\]
must use $\Omega(1/\epsilon)$ bits of memory, even for instances with $N=\Theta(1/\epsilon)$.
\end{theorem}

\begin{proof}
We again reduce from \textsc{Index} (Lemma~\ref{lem:index}).
\[
n=\left\lfloor\frac{1}{100\epsilon}\right\rfloor,\qquad
N=2n+1,\qquad
B=n,\qquad
T=4n^2.
\]
We will construct a stream over $x\in[N]$ with labels in $\{-1,+1\}$.

\textbf{Alice's prefix.}
Given $z\in\{0,1\}^n$, for each $k\in[n]$ Alice inserts exactly $B$ copies of
\[
(x,y)=(2k,\; s_k),\qquad\text{where } s_k:=\begin{cases}+1 &\text{if } z_k=1,\\ -1 &\text{if } z_k=0.\end{cases}
\]
Thus the number of Alice points is $m_A=nB=n^2$.
Alice runs the one-pass algorithm~$\mathcal{A}$ on this prefix and sends its memory state ($s$ bits) to Bob.

\textbf{Bob's suffix.}
Given index $i\in[n]$, Bob appends
\[
T\ \text{copies of }(x,y)=(2i-1,\,-1)
\qquad\text{and}\qquad
T\ \text{copies of }(x,y)=(2i+1,\,+1).
\]
Let the full stream be $S=S_A(z)\circ S_B(i)$, whose total length is
\[
m=m_A+2T=n^2+8n^2=9n^2.
\]
Define two adjacent split positions
\[
j^-:=2i-1,\qquad j^+:=2i.
\]
Note that $j^-$ places the $x=2i$ block on the right, while $j^+$ places it on the left.

\begin{figure}[t]
\centering
\begin{tikzpicture}[x=1.3cm,y=1cm]
	\draw[->,thick] (0,0) -- (6,0) node[right] {$x$};
	\draw (1,0.08) -- (1,-0.08) node[below=4pt] {$2i\!-\!1$};
	\draw (3,0.08) -- (3,-0.08) node[below=4pt] {$2i$};
	\draw (5,0.08) -- (5,-0.08) node[below=4pt] {$2i\!+\!1$};

	\node[draw=black,fill=white,align=center,font=\scriptsize,text width=1.5cm,minimum height=0.75cm] at (1,0.85) {$T$ copies\\$(y=-1)$};
	\node[draw=black,fill=white,align=center,font=\scriptsize,text width=1.8cm,minimum height=0.75cm] at (3,0.85) {$G$: $B$ copies\\$(y=s_i)$};
	\node[draw=black,fill=white,align=center,font=\scriptsize,text width=1.5cm,minimum height=0.75cm] at (5,0.85) {$T$ copies\\$(y=+1)$};

	\draw[dashed,thick] (1,-0.2) -- (1,0.45);
	\draw[dashed,thick] (1,1.25) -- (1,1.5) node[above] {$j^-\!=\!2i\!-\!1$};
	\draw[dashed,thick] (3,-0.2) -- (3,0.45);
	\draw[dashed,thick] (3,1.25) -- (3,1.5) node[above] {$j^+\!=\!2i$};
\end{tikzpicture}
\caption{Bob contributes two large anchor blocks at $x=2i-1$ (label $-1$) and $x=2i+1$ (label $+1$),
while Alice contributes the middle block~$G$ at $x=2i$ with label $s_i\in\{-1,+1\}$.
The competitive splits are the two adjacent candidates $j^-$ and $j^+$;
knowing the better split reveals the hidden bit~$z_i$.}
\label{fig:lb-classification-index-diagram}
\end{figure}

\begin{claim}\label{claim:far-splits-classification}
For every $j\notin\{j^-,j^+\}$, we have $L_{\rm mis}(j)\ge 4/9$.
\end{claim}

\begin{proof}
Consider any $j\le 2i-2$.
Then both anchor blocks $(2i-1,-1)$ and $(2i+1,+1)$ lie on the right side $[j+1,N]$. Hence
\[
f_{-1,[j+1,N]}\ge T,\qquad f_{+1,[j+1,N]}\ge T,
\]
so $\min\{f_{-1,[j+1,N]},f_{+1,[j+1,N]}\}\ge T$ and therefore
\[
L_{\rm mis}(j)\ \ge\ \frac{T}{m}\ =\ \frac{4n^2}{9n^2}\ =\ \frac{4}{9}.
\]
A symmetric argument applies to any $j\ge 2i+1$: then both anchors lie on the left side $[1,j]$,
giving $\min\{f_{-1,[1,j]},f_{+1,[1,j]}\}\ge T$ and again $L_{\rm mis}(j)\ge 4/9$.
\end{proof}

\begin{claim}\label{claim:middle-splits-classification}
For each $j\in\{j^-,j^+\}$, we have $L_{\rm mis}(j)\le 1/9$.
Moreover, the left majority label is $-1$ and the right majority label is $+1$.
Consequently,
\[
\opt=\min\{L_{\rm mis}(j^-),L_{\rm mis}(j^+)\}.
\]
\end{claim}

\begin{proof}
Fix $j\in\{j^-,j^+\}$.
Then all $T$ copies of $(2i-1,-1)$ lie on the left side, and all $T$ copies of $(2i+1,+1)$
lie on the right side.

Since $T=4n^2>n^2=m_A$, the left side contains more $-1$ labels from Bob than the total number
of Alice points, so the left majority label is $-1$.
Similarly, the right majority label is $+1$.

Thus, under the optimal majority vote on each side, all Bob points are classified correctly.
Only Alice's points can be misclassified, and there are exactly $m_A=n^2$ such points.
Therefore
\[
L_{\rm mis}(j)\le \frac{m_A}{m}=\frac{n^2}{9n^2}=\frac19.
\]

By Claim~\ref{claim:far-splits-classification}, every $j\notin\{j^-,j^+\}$ satisfies
\[
L_{\rm mis}(j)\ge \frac49>\frac19,
\]
so the optimum is attained at one of $j^-,j^+$.
\end{proof}

\begin{claim}\label{claim:gap-classification}
$|L_{\rm mis}(j^-)-L_{\rm mis}(j^+)|>\epsilon$, and $\arg\min\{L_{\rm mis}(j^-),L_{\rm mis}(j^+)\}$ reveals $z_i$.
\end{claim}

\begin{proof}
Let $G$ be the block at $x=2i$, consisting of $B=n$ copies of label $s_i\in\{-1,+1\}$
(where $s_i=+1$ iff $z_i=1$, and $s_i=-1$ iff $z_i=0$).
All points other than $G$ lie on the same side under both $j^-$ and $j^+$, and
Claim~\ref{claim:middle-splits-classification} shows the side-majorities remain $-1$ on the left
and $+1$ on the right.
Thus the only change in misclassification count comes from moving $G$ from right (under $j^-$)
to left (under $j^+$).

{\bf Case 1: $z_i=1$} (so $s_i=+1$).
Under $j^-$, $G$ lies on the right (majority $+1$), contributing $0$ to $f_{-1,(j^-,N]}$.
Under $j^+$, $G$ lies on the left (majority $-1$), contributing $B$ to $f_{+1,[1,j^+]}$.
Hence $L_{\rm mis}(j^+)=L_{\rm mis}(j^-)+B/m$.

{\bf Case 2: $z_i=0$} (so $s_i=-1$).
Under $j^-$, $G$ lies on the right (majority $+1$), contributing $B$ to $f_{-1,(j^-,N]}$.
Under $j^+$, $G$ lies on the left (majority $-1$), contributing $0$ to $f_{+1,[1,j^+]}$.
Hence $L_{\rm mis}(j^-)=L_{\rm mis}(j^+)+B/m$.

In both cases $|L_{\rm mis}(j^-)-L_{\rm mis}(j^+)|=B/m=1/(9n)\ge 10\epsilon>\epsilon$, and the minimizer reveals $z_i$.
\end{proof}

Since $\opt=\min\{L_{\rm mis}(j^-),L_{\rm mis}(j^+)\}$ and $|L_{\rm mis}(j^-)-L_{\rm mis}(j^+)|>\epsilon$,
any algorithm outputting $\widehat{j}$ with $L_{\rm mis}(\widehat{j})\le \opt+\epsilon$ must output the correct
minimizer in $\{j^-,j^+\}$ with probability at least~$2/3$, enabling Bob to recover $z_i$.
Thus, the memory state ($s$ bits) yields a one-way protocol for \textsc{Index} with communication~$s$,
so by Lemma~\ref{lem:index}, $s=\Omega(n)=\Omega(1/\epsilon)$.
\end{proof}

\subsection{Lower bound for classification under Gini impurity}\label{subsec:lb-gini}

Using the re-encoding trick, we can easily prove a space lower bound for the loss based on the Gini impurity.

\begin{theorem}\label{thm:lb-gini}
Fix $\epsilon\in(0,10^{-3})$.
Any one-pass randomized streaming algorithm that, with probability at least $2/3$,
outputs a split $\widehat{j}\in[N]$ satisfying
\[
L_{\mathrm{Gini}}(\widehat{j})\le \opt+\epsilon
\]
must use $\Omega(1/\epsilon)$ bits of memory, even for instances with $N=\Theta(1/\epsilon)$.
\end{theorem}

\bibliography{lipics-v2021-sample-article}

\appendix

\section{Omitted Proofs} \label{sec:omitted}

\begin{proof}[Proof of Lemma \ref{lem:split-shift}]
Let
\[
A:=\{i:x_i\le j\},\qquad
B:=\{i:j<x_i\le j'\},\qquad
C:=\{i:x_i>j'\}.
\]

Then $|B|=b$. Thus split $j$ corresponds to the partition
\[
(A,\;B\cup C),
\]
while split $j'$ corresponds to
\[
(A\cup B,\;C).
\]

Let $\alpha$ and $\beta$ be values (one can show that $\alpha = \mu(j)$ and $\beta=\gamma(j)$) such that
\[
L_{\mathrm{MSE}}(j)=\frac1m\left(\sum_{i\in A}(y_i-\alpha)^2+\sum_{i\in B\cup C}(y_i-\beta)^2\right).
\]
Since $L_{\mathrm{MSE}}(j')$ is the minimum possible loss for split $j'$, we may upper bound it by using the
same constants $\alpha$ and $\beta$:
\[
L_{\mathrm{MSE}}(j')
\le
\frac1m\left(\sum_{i\in A\cup B}(y_i-\alpha)^2+\sum_{i\in C}(y_i-\beta)^2\right).
\]
Subtracting the two expressions gives
\[
L_{\mathrm{MSE}}(j')-L_{\mathrm{MSE}}(j)
\le
\frac1m\sum_{i\in B}\Big((y_i-\alpha)^2-(y_i-\beta)^2\Big).
\]
Because $y_i,\alpha,\beta\in[0,M]$, each term lies in $[-M^2,M^2]$, so each summand is at most
$M^2$. Hence
\[
L_{\mathrm{MSE}}(j')-L_{\mathrm{MSE}}(j)\le \frac{bM^2}{m}.
\]

Now reverse the roles of $j$ and $j'$. Let $(\alpha',\beta')$ be an optimal pair for split $j'$. Using
$(\alpha',\beta')$ as a feasible choice for split $j$, the same argument gives
\[
L_{\mathrm{MSE}}(j)-L_{\mathrm{MSE}}(j')\le \frac{bM^2}{m}.
\]

Combining the two bounds,
\[
|L_{\mathrm{MSE}}(j')-L_{\mathrm{MSE}}(j)|\le \frac{bM^2}{m}.
\]
\end{proof}

\begin{proof}[Proof of Lemma \ref{lem:heavy-hit}]
If $K\ge m$, then the reservoir contains all stream items, so the claim is trivial.

Assume $K<m$. Fix an interval $(a,b]\subseteq[N]$, and let $t$ be the number of stream items
whose feature value lies in $(a,b]$. The final reservoir
is a uniformly random $K$-subset of the $m$ stream items. Hence
\begin{align*}
\Pr[\text{the reservoir misses all data points in }(a,b]]
&= \frac{\binom{m-t}{K}}{\binom{m}{K}}
= \prod_{i=0}^{K-1}\frac{m-t-i}{m-i} \\
&\le \prod_{i=0}^{K-1}\left(1-\frac{t}{m}\right)
= \left(1-\frac{t}{m}\right)^K,
\end{align*}
where each factor uses $\frac{m-t-i}{m-i} = 1-\frac{t}{m-i}\le 1-\frac{t}{m}$ since $m-i\le m$.
If $t>\tau m$, then
\[
\Pr[\text{the reservoir misses all data points in }(a,b]]
\le
(1-\tau)^K
\le
e^{-K\tau}.
\]
Since $K=\lceil C\log N/\tau\rceil$, this is at most $N^{-C}$.

There are at most $N^2$ intervals of the form $(a,b]\subseteq[N]$. By a union bound,
the probability that some interval containing more than $\tau m$ points is missed by the reservoir
is at most
\[
N^2\cdot N^{-C}.
\]
Choosing $C\ge 4$ makes this at most $1/N^2$.
\end{proof}

\begin{proof}[Proof of Lemma \ref{lem:mean-min}]
Write $(y_t-a)=(y_t-\bar y)+(\bar y-a)$ and expand:
\[
\sum_{t=1}^k (y_t-a)^2
=\sum_{t=1}^k (y_t-\bar y)^2
+2(\bar y-a)\sum_{t=1}^k (y_t-\bar y)
+k(\bar y-a)^2.
\]
The cross term vanishes since $\sum_{t=1}^k(y_t-\bar y)=0$, giving the identity. The inequality follows from $k(a-\bar y)^2\ge 0$.
\end{proof}

\begin{proof}[Proof of Lemma~\ref{lem:gini-shift}]
For a fixed split $j$, let $L := [1,j]$ and $R := [j+1,N]$. Re-encode each label $y_i\in\{-1,+1\}$ as $z_i=(y_i+1)/2\in\{0,1\}$, and define
\[
L_{\mathrm{MSE}}(j) := \frac{1}{m}\!\left(\sum_{i:\,x_i\le j}(z_i-\bar z_L)^2 + \sum_{i:\,x_i>j}(z_i-\bar z_R)^2\right),
\]
where $\bar z_L$ and $\bar z_R$ are the means of the $z_i$'s on the left and right sides, respectively.

Let
\[
n_L := f_{+1,L}+f_{-1,L}, \qquad n_R := f_{+1,R}+f_{-1,R}.
\]
If $n_L>0$, then $\bar z_L = f_{+1,L}/n_L$, and
\begin{align*}
	\sum_{i:\,x_i\le j}(z_i-\bar z_L)^2
	&= f_{+1,L}(1-\bar z_L)^2 + f_{-1,L}(0-\bar z_L)^2 \\
	&= f_{+1,L}\left(\frac{f_{-1,L}}{n_L}\right)^2 + f_{-1,L}\left(\frac{f_{+1,L}}{n_L}\right)^2 \\
	&= \frac{f_{+1,L}f_{-1,L}^2 + f_{-1,L}f_{+1,L}^2}{n_L^2} \\
	&= \frac{f_{+1,L}f_{-1,L}(f_{+1,L}+f_{-1,L})}{n_L^2} = \frac{f_{+1,L}f_{-1,L}}{n_L}.
\end{align*}
If $n_L=0$, then the left side is empty and this contribution is $0$. Likewise, the right-side contribution equals $f_{+1,R}f_{-1,R}/n_R$ when $n_R>0$, and $0$ when $n_R=0$.

Hence
\[
2L_{\mathrm{MSE}}(j)
=
\frac{2}{m}\left(\frac{f_{+1,L}f_{-1,L}}{n_L}+\frac{f_{+1,R}f_{-1,R}}{n_R}\right)
=
L_{\mathrm{Gini}}(j),
\]
where empty-side terms are interpreted as $0$. 
Since $z_i\in[0,1]$ (i.e.\ $M=1$), Lemma~\ref{lem:split-shift} gives $|L_{\mathrm{MSE}}(j')-L_{\mathrm{MSE}}(j)|\le \ell/m$, and therefore
\[
\bigl|L_{\mathrm{Gini}}(j')-L_{\mathrm{Gini}}(j)\bigr| = 2\bigl|L_{\mathrm{MSE}}(j')-L_{\mathrm{MSE}}(j)\bigr| \le \frac{2\ell}{m}.\qedhere
\]
\end{proof}

\begin{proof}[Proof of Lemma~\ref{lem:add-block}]
Let $\mu' := \frac{a\mu+Bv}{a+B}$ be the mean of $D\cup G$. Apply Lemma~\ref{lem:mean-min} to $D$,
\[
\sum_{y\in D}(y-\mu')^2 = \SSE(D) + a(\mu-\mu')^2.
\]
Since every element of $G$ equals $v$, we have $\sum_{y\in G}(y-\mu')^2=B(v-\mu')^2$. Summing,
\[
\SSE(D\cup G)
=\sum_{y\in D\cup G}(y-\mu')^2
=\SSE(D)+a(\mu-\mu')^2 + B(v-\mu')^2.
\]
It remains to show $a(\mu-\mu')^2+B(v-\mu')^2=\frac{aB}{a+B}(\mu-v)^2$. Substituting $\mu-\mu'=\frac{B(\mu-v)}{a+B}$ and $v-\mu'=\frac{a(v-\mu)}{a+B}$,
\begin{align*}
a(\mu-\mu')^2+B(v-\mu')^2
= \frac{aB^2(\mu-v)^2}{(a+B)^2}+\frac{Ba^2(\mu-v)^2}{(a+B)^2}
& = \frac{aB(a+B)(\mu-v)^2}{(a+B)^2} \\
& = \frac{aB}{a+B}(\mu-v)^2. \qedhere
\end{align*}
\end{proof}

\begin{proof}[Proof of Theorem \ref{thm:lb-gini}]
Re-encode each label $y_i\in\{-1,+1\}$ as $z_i=(y_i+1)/2\in\{0,1\}$. By the identity established in the proof of Lemma~\ref{lem:gini-shift}, we have
\[
L_{\mathrm{Gini}}(j)=2\,L_{\mathrm{MSE}}(j)
\]
for every split $j$, where $L_{\mathrm{MSE}}$ is the squared loss on the re-encoded stream with $M=1$.
Hence any one-pass algorithm that outputs $\widehat{j}$ with $L_{\mathrm{Gini}}(\widehat{j})\le\opt+\epsilon$ also satisfies $L_{\mathrm{MSE}}(\widehat{j})\le\OPT_{\mathrm{MSE}}+\epsilon/2$, giving a one-pass $(\epsilon/2)$-additive algorithm for regression with $M=1$.
By Theorem~\ref{thm:lb-regression} with $M=1$, such an algorithm requires $\Omega(M^2/(\epsilon/2))=\Omega(1/\epsilon)$ bits of memory.
\end{proof}

\end{document}